\pdfoutput=1 
\documentclass[a4paper,USenglish,cleveref]{lipics-v2021}

\hideLIPIcs  

\usepackage{amsfonts,amsmath,amsthm,bm}
\usepackage{graphicx}
\usepackage{tikz,pgf}
\usetikzlibrary{arrows.meta}
\usetikzlibrary{cd}
\usetikzlibrary{fit,shapes.geometric}
\usetikzlibrary{backgrounds}
\usepackage{todonotes}
\usepackage{smp}
\usepackage[appendix=append]{apxproof} 
\usepackage[append]{optional}
\usepackage[linesnumbered,algoruled,noend,nokwfunc,procnumbered]{algorithm2e}
\usepackage{environ}

\makeatletter
\newsavebox{\measure@tikzpicture}
\NewEnviron{scaletikzpicturetowidth}[1]{%
  \def\tikz@width{#1}%
  \begin{lrbox}{\measure@tikzpicture}%
  \BODY
  \end{lrbox}%
  \pgfmathparse{#1/\wd\measure@tikzpicture}%
  \BODY
}
\makeatother

\newtheoremrep{theorem}{Theorem}[section]
\newtheoremrep{lemma}{Lemma}[section]
\newtheoremrep{corollary}{Corollary}[section]
\newtheoremrep{proposition}{Proposition}[section]

\definecolor{lightestgray}{gray}{0.9}

\AtBeginEnvironment{procedure}{\crefalias{algocf}{procedure}}
\crefname{procedure}{procedure}{procedures}
\Crefname{procedure}{Procedure}{Procedures}

\SetKwProg{Fn}{}{ :}{}
\SetKwData{conin}{x}
\SetKwData{conout}{y}

\title{Stabilizing Consensus is Impossible in Lossy Iterated Immediate Snapshot Models}

\author{Stephan Felber $\dag$}{Vienna University of Technology, Austria}{stephan.felber@tuwien.ac.at}{0009-0003-6576-1468}{}

\author{Hugo {Rincon Galeana $\ddag$}}{Vienna University of Technology, Austria\\ Berlin University of Technology, Germany}{hugorincongaleana@gmail.com}{0000-0002-8152-1275}{}

\authorrunning{S. Felber and H. Rincon Galeana} 

\Copyright{Stephan Felber and Hugo Rincon Galeana} 



\category{} 

\relatedversion{} 

\acknowledgements{
This work was supported by the DMAC project (P32431) $\dag, \ddag$; supported by German Research Foundation (DFG), Schwerpunktprogramm (SPP 2378), project ReNO, 2023-2027 \ddag; and the financial support of Najla Amira Ochoa Leonor \ddag. We would like to thank Ulrich Schmid and Kyrill Winkler for their helpful remarks, which undoubtedly improved this paper. 
}

\nolinenumbers 

\EventEditors{Silvia Bonomi, Letterio Galletta, Etienne Rivi\`{e}re, and Valerio Schiavoni}
\EventNoEds{4}
\EventLongTitle{28th International Conference on Principles of Distributed Systems (OPODIS 2024)}
\EventShortTitle{OPODIS 2024}
\EventAcronym{OPODIS}
\EventYear{2024}
\EventDate{December 11--13, 2024}
\EventLocation{Lucca, Italy}
\EventLogo{}
\SeriesVolume{324}
\ArticleNo{6}

\begin{document}

\maketitle

\begin{abstract}
    A substantial portion of distributed computing research is dedicated to terminating problems like consensus and similar agreement problems. However, non-terminating problems have been intensively studied in the context of self-stabilizing distributed algorithms, where processes may start from arbitrary initial states and can tolerate arbitrary transient faults. In between lie stabilizing problems, where the processes start from a well-defined initial state, but do not need to decide irrevocably and are allowed to change their decision finitely often until a stable decision is eventually reached.

    Stabilizing consensus has been studied within the context of synchronous message adversaries. In particular, Charron-Bost and Moran showed that a necessary condition for stabilizing consensus is the existence of at least one process that reaches all others infinitely often (a perpetual broadcaster). However, it was left open whether this is also a sufficient condition for solving stabilizing consensus.

    In this paper, we introduce the novel Delayed Lossy-Link (\DLL) model, and the Lossy Iterated Immediate Snapshot Model (\LIIS), for which we show stabilizing consensus to be impossible. The \DLL model is introduced as a variant of the well-known Lossy-Link model, which admits silence periods of arbitrary but finite length. The \LIIS model is a variant of the Iterated Immediate Snapshot (\IIS), model which admits finite length periods of at most $f$ omission faults per layer. In particular, we show that stabilizing consensus is impossible even when $f=1$. 
    
    Our results show that even in a model with very strong connectivity, namely, the Iterated Immediate Snapshot (\IIS) model, a single omission fault per layer effectively disables stabilizing consensus. Furthermore, since the \DLL model always has a perpetual broadcaster, the mere existence of a perpetual broadcaster, even in a crash-free setting, is not sufficient for solving stabilizing consensus, negatively answering the open question posed by Charron-Bost and Moran.
    
\end{abstract}

\section{Introduction}
\label{sec:Intro}

Agreement tasks, and in particular consensus, have always been a focal point of distributed computing research, not only because of their practical applicability, but also because consensus tasks characterize very precisely the limits of a distributed computing model. Understanding such limitations not only facilitates the quick assessment of problem solvability in a given model, but also sheds light on the impact that certain properties, such as synchrony, impose on the system. For instance, the celebrated FLP impossibility result~\cite{FLP/ACM} reveals the devastating effect of both asynchrony and process crashes on deterministic consensus solvability. In a similar vein, it has been shown that process crashes impair general task solvability~\cite{HS94} even in asynchronous shared memory systems, and that byzantine faults~\cite{LSP82} and/or message loss~\cite{SW89,SWK09,CGP15,WPRSS23,winkler_characterization_2020} impair consensus solvability even in the context of synchronous message passing systems.

Although distributed systems research has focused extensively on \emph{terminating} tasks, there exist interesting applications and problems that are inherently \emph{non-terminating}. Apart from the specific class of self-stabilizing distributed algorithms \cite{Dij74,Dol00}, which can even tolerate massive transient faults, asymptotic consensus~\cite{BT89,BHOT05}, stabilizing consensus~\cite{AFJ06,CM19:DC,SS21:SSS}, and approximate consensus~\cite{CFN15:ICALP,DLPSW86} are examples of non-terminating distributed tasks: processes are not required to terminate after having computed some final value, but rather eventually converge to some stable configuration. Such tasks are not only of theoretical interest, but are also essential for implementing practical distributed problems such as clock synchronization~\cite{LL84:PODC, WS07:DC} and sensor fusion~\cite{BS92}.

In this paper, we focus on stabilizing consensus, which can be viewed as the non-terminating variant of consensus. As in the case of consensus, all processes start with their own input values and must \emph{eventually} agree on a common value. The fundamental difference to terminating consensus is that every process may start arbitrarily late, and does not need to decide on a value irrevocably and exactly once, but can change its decision value \emph{finitely often}. Nevertheless, all processes must eventually stabilize on the \emph{same} decision value.

It should be noted that stabilizing consensus is also loosely related to asymptotic consensus ~\cite{FNS21}, in the sense that, in neither problem, processes are required to terminate. A fundamental difference, however, is that asymptotic consensus allows processes to decide from a real-valued range, in contrast to the validity condition for stabilizing consensus, which only allows decisions from the pool of input values. Stabilizing consensus is a stronger problem, since any protocol that solves stabilizing consensus also solves asymptotic consensus. Indeed, we provide a simple novel model introduced below (the \DLL), where stabilizing consensus is impossible (see~\cref{thm:stabagreeimp}), yet asymptotic consensus can be solved (see~\cref{thm:asympdll}, corresponding to~\Cref{proc:asympDLL}).
In this light, asymptotic consensus is more closely related to approximate consensus~\cite{CFN15:ICALP}, while stabilizing consensus is more closely related to consensus. 

\subsection{Related work}
\label{subsec:relwork}
We study stabilizing consensus in the \emph{Synchronous Message Passing} setting where processes communicate through uni-directional links in a round by round fashion and an \emph{adversary} suppresses certain links in every round. More specifically, the communication in every round is modeled as a directed graph sequence, which is under the control of a \emph{message adversary}~\cite{AG13}. Charron-Bost and Moran~\cite{CM19:DC} provided a class of algorithms in this setting, called the \emph{MinMax Algorithms}, solving stabilizing consensus under any message adversary that generates graph sequences adhering to two constraints: (i) there exists a process (a perpetual broadcaster) that is able to reach all other process (possibly via multiple hops) infinitely often, and (ii) the broadcasting time (in terms of rounds) for doing so is bounded by an unknown constant. The authors showed that (i) is necessary for solving stabilizing consensus, but leave the question open whether or not (i) is also sufficient. 

In this paper, we extend the Lossy-Link model, introduced by Santoro and Widmayer~\cite{SW89}, where two processes are connected by a pair of links that may drop at most one (directional) message per round, to the novel \emph{Delayed Lossy-Link} (\DLL) model, where both messages may be lost but (some) messages are guaranteed to be transmitted infinitely often. Whereas stabilizing consensus can be solved in the Lossy-Link model, we show that stabilizing consensus is impossible in the \DLL model, thus negatively answering Charron-Bost and Moran's open question in \cref{thm:stabimpdiis}. Our proof vaguely resembles the bivalency proof used for showing the impossibility of terminating consensus in the Lossy-Link model~\cite{SW89}, in the sense that we construct a forever \emph{conflicted} run, i.e., a run where processes always decide on invalid output configurations, by extending conflicted prefixes. Whereas a bivalent prefix leads to more than one decision value eventually, a conflicted prefix already yields an output configuration that violates agreement. A bivalency argument is not enough to prove stabilizing consensus impossible, however.

Afek and Gafni~\cite{AG13} showed that asynchronous wait-free shared memory models such as the \emph{Iterated Immediate Snapshot} (\IIS) model can also be represented by a synchronous message adversary. We extend the message adversary formulation of the \IIS model by adding up to $f$ read-omission faults per round to the snapshot operations. We call this new model the \emph{Lossy Iterated Immediate Snapshot model} $\LIIS(f)$. In particular, in the case of only 2 processes, $\LIIS(1)$ corresponds to the \DLL model. The $\LIIS(f)$ model is loosely related to the $d$-solo models studied by Herlihy, Rajsbaum, Raynal and Stainer~\cite{HRRS17}, as the presence of read omissions may lead to $d$-solo rounds, i.e., rounds where up to $d$ processes do not receive any information from the rest of the processes.

\begin{toappendix}
\begin{theorem}[\DLL allows Asymptotic consensus]\label{thm:asympdll}
    Asymptotic consensus, where processes are only required to agree in the limit, is solvable in the \DLL.
\end{theorem}

\begin{proof}
    Consider the following implementation of a well-known averaging algorithm. Note that in \cref{proc:asympDLL} we intentionally resort to a simple pseudocode representation (instead of the model introduced in \cref{sec:messadv}) for readability purposes.
    
    \begin{algorithm}
        \caption{Averaging algorithm adapted to the \DLL. This algorithm solves asymptotic consensus in the \DLL. Code for process p.}\label{proc:asympDLL}
        \DontPrintSemicolon
        $current = \conin_p$\;
        \ForEach{Round $r = 1, 2, 3, ...$}{
            $\texttt{send}(current)$\;
            $values = \texttt{receive}()$\;
            $current = \texttt{avg}(values)$\;
        }
    \end{algorithm}
    
    Consider a communication pattern $\sigma \in \DLL$ and observe the following cases.
  \begin{itemize}
       \item[1] Assume $\sigma(r) = \tpqp$ for some $r$: then both \tp, \tq receive each others value and compute the same average. As they no both hold the same value in $current$, they will never change and have converged with $\epsilon = 0$.
        \item[2] Assume $\sigma(r) \neq \tpqp$ for all $r$: then, w.l.o.g., only \tq receives messages from \tp, so $\sigma(r) \in \{ \tpq, \tnone \}$. Now, every time \tq receives \tp's value, it will halve the distance between the two values. As communication liveness is guaranteed, any required $\epsilon > |current_\tp - current_\tq|$ will be achieved.
    \end{itemize}
\end{proof}
\end{toappendix}

\subsection{Contributions and paper organization}
\label{subsec:Contr}
 In \cref{sec:messadv}, we provide the preliminary definitions and the overall framework that are essential for our results. In \cref{sec:lldll}, we present the Lossy-Link model, and introduce the Delayed Lossy-Link model for which we prove stabilizing consensus to be impossible, namely \cref{thm:stabagreeimp}. This is our first main result and constitutes the foundation for the rest of the paper. In \cref{sec:iisdiis}, we present the \IIS model and our novel $\LIIS(f)$ model, and show that stabilizing consensus is not solvable in $\LIIS(f)$ for any $f \geq 1$, via \cref{thm:stabimpdiis}. Finally, in \cref{sec:conclusion} we discuss the implications of our result and provide some interesting directions for future work.

\section{The System Model}
\label{sec:messadv}

We consider a finite set of process $\Pi = \{p_1, \ldots, p_n \}$ that communicate through message passing via directional point-to-point links in lock-step synchronous rounds. We assume that rounds are communication-closed, i.e., messages sent at a round $r$ will also arrive at round $r$ or will not be delivered. Note that in contrast to \cite{CM19:DC} we assume synchronous process starts (i.e., starting round $s_p=1$ for all $p$) as this special case already facilitates our impossibility result. Obviously, our result then carries over to the more general model presented in the original paper.

Each process $p_i$ is given an initial input value $\conin_i$ from a finite set of inputs $I$. We say that an initial global configuration is an $n$-tuple $(p_i,\conin_i)_{i=1}^n$ of process-input pairs. We denote the set of possible \emph{input configurations} as $\mathcal{I}$. Similarly, an \emph{output configuration} is an $n$-tuple $(p_i, \conout_i)_{i=1}^n$. We denote the set of possible output configurations by $\mathcal{O}$.

We model distributed problems as input/output relations called \emph{tasks} that map a set of input configurations $\mathcal{I}$ to a set of valid output configurations. Formally, a task is a triple $T=\langle \mathcal{I}, \mathcal{O}, \task \rangle$ where $\task: \mathcal{I} \rightarrow 2^{\mathcal{O}}$ is a validity map that determines a subset of valid output configurations from a single input configuration. Tasks are widely used in the literature to capture problems in distributed setting, see for example~\cite{RS13:PODC, gafni_distributed_2010, galeana_continuous_2022}.

The communication is governed by a message adversary that determines which processes are able to communicate at a given round. We use a directed graph $G_r$ on $\Pi$, which we call the \emph{communication graph}, to represent the message exchange at round $r$. We denote by $In_G(p_i)$ the in-neighborhood of $p_i$ in graph $G$, which always includes $p_i$. Note that $In_{G_r}(p_i)$ represents the set of processes from which $p_i$ receives a message at round $r$. We represent the communication in an execution by a communication graph sequence $\sigma = (G_i)_{i=1}^\infty$, called a \emph{communication pattern}, where $\sigma(k) = G_k$. We denote by $G \comp \sigma$ the concatenation of $G$ with $\sigma$ and by $G^k$ the sequence of $G$ repeated $k$ times. If $\sigma = (G_i)_{i=1}^\infty$ is a communication pattern, we say that $(G_i)_{i=1}^k$ is the $k$-\emph{prefix} of $\sigma$, denoted by $\sigma|_1^k$. Conversely, $\sigma|_k^\infty = (G_i)_{i=k}^\infty$ denotes the suffix starting from round $k$. We define $\mu = (H_i)_{i=1}^\infty$ to be a sub-sequence of $\sigma = (G_i)_{i=1}^\infty$, denoted by $\mu \sqsubseteq \sigma$, iff there is a strictly increasing sequence $s: \N \rightarrow \mathbb{N}$ such that $H_i = G_{s(i)}$. Note that the sub-sequence relation $\sqsubseteq$ is a partial order on sequences with the same domain and co-domain. A \emph{message adversary} is just a set of communication patterns $\mathcal{M}$.

Since our focus is on solvability, we assume for convenience that the processes execute a \emph{full-information} \emph{flooding protocol}. That is, at any given round $r$, each process will send a copy of its local state, also called \emph{view} to all possible recipients. For terminating protocols and terminating tasks, we consider strong termination, i.e., a process stops its execution whenever a decision is made by the protocol. Thus, processes will stop sending messages and updating their views once the protocol has output a decision value. In contrast, for stabilizing tasks and protocols, we consider that the protocol never terminates and processes keep communicating and updating their views forever.

We define the \emph{view} of a process $p_i$ at the end of a round $r$ under $\sigma = (G_i)_{i=1}^\infty$ and initial configuration $\config[0]$ as: 
$view_{\config[0] : \sigma}(p_i,r) := \langle V(\config[0] : \sigma, p_i, r), p_i, r \rangle$, where $V(\config[0]:\sigma, p_i, r) = \{ view_{\config[0]:\sigma}(p_j, r-1) \; \vert \; p_j \in In_{\sigma_{r-1}}(p_i) \}$. We define $view_{\config[0]:\sigma}(p_i, 0) := \langle \{ \conin_i \}, p_i, 0 \rangle$, where $\conin_i$ is the input value of process $p_i$ in $\config[0]$.

We also employ a \emph{subview} relation, denoted by $view_{\config[0]:\sigma}(p_i, r) \prec view_{\config[0]:\sigma}(p_j, s)$, which can be defined inductively: $view_{\config[0]:\sigma}(p_j, s) \prec view_{\config[0]\sigma}(p_i,r)$ iff $s = r-1$ and  $view_{\config[0]:\sigma}(p_j, s) \in V(\config[0]:\sigma, p_i, r)$; or $view_{\config[0]_\sigma}(p_j, s) \prec view_{\config[0]:\sigma}(p_k, r-1)$ for some \\ $view_{\config[0]:\sigma}(p_k, r-1) \in V(\config[0]:\sigma, p_i, r)$. Intuitively, this subview relation captures all the previous local states that a given process $p_i$ is aware of at a given round $r$.

We denote by $\leafs(view_{\config[0]:\sigma}(p_i,r)) := \{ \conin_j \; \vert \; \langle \{ \conin_j\},p_j, 0 \rangle \prec view_{\config[0]:\sigma}(p_i,r) \}$ the set of inputs $p_i$ has heard of by round $r$. Whenever $\config[0]$ is clear from the context or not relevant, we will omit it in favor of notational simplicity.

We define the kernel $Ker(\sigma)$ of a sequence $\sigma$ as the set of processes that reach all other processes infinitely often, i.e., processes in $Ker(\sigma)$ are able to disseminate their local state at any point in the run. Formally, $p \in Ker(\sigma)$ if and only if $\forall r>0 : \exists r'>r : \forall q\in\Pi : view_\sigma(p,r)\prec view_\sigma(q,r')$. 

We define the \emph{global configuration} at the end of round $r$ under a communication pattern $\sigma$ as an $n$-tuple $\mathcal{C}^r_{\sigma} := (view_{\sigma}(p_i,r))_{i=1}^n$.

As we only consider deterministic protocols, global configurations are fully determined by the input configuration, and the communication pattern $\sigma$. Therefore, in the context of this paper, we will consider that a run consists of the input configuration and communication pattern. Furthermore, whenever it is clear from the context, or the input configuration is not relevant, we will use communication pattern and run interchangeably. 

We use views to determine indistinguishability between global configurations. More precisely, a global configuration $\mathcal{C}^r_{\alpha: \sigma}$ is \emph{indistinguishable} from a global configuration $\mathcal{C}^r_{\beta: \mu}$ for a process $p_i \in \Pi$ iff $view_{\alpha:\sigma}(p_i,r) = view_{\beta: \mu}(p_i,r)$, where $\alpha, \beta \in \mathcal{I}$ are input configurations, and $\sigma, \mu \in \mathcal{M}$ are communication patterns.

Note that any protocol can be split into a full-information flooding part that generates views and sends messages, and a decision map that produces an output value. If the output values form a valid configuration, then we say that the protocol solves a problem.

Formally, we define a \emph{decision map} of a protocol \Prot as a function that maps process views generated by the full information flooding protocol to output values. We will denote a decision map as $\delta_{\Prot}: Views(\mathcal{M}) \rightarrow O$, where $Views(\mathcal{M})$ is the set of possible views induced by the message adversary $\mathcal{M}$ under the protocol \Prot, and $O$ is the set of possible output values. Note that for terminating protocols, an additional undecided value $\bot$ may be used for representing that the protocol is not yet ready to decide. However, since we are focusing only on stabilizing protocols, we require that a decision is made at every view.

We define the \emph{stabilization} property for tasks, which can be thought of as a relaxation of termination. Whereas termination requires a process to irrevocably choose its output value once and possibly even finalize its execution, stabilization allows the protocol to run indefinitely and correct its output finitely often.

Nevertheless processes must eventually \emph{stabilize} on a common constant output value. We call the round on which a process $p_i$ stabilizes, i.e., never changes its output again, the \emph{stabilization round} $s_i$. It should be noted that stabilization may be reached agnostically. This means that there might be protocols where processes are guaranteed to stabilize, yet they are possibly never aware that they have already reached a stable state.

We model stabilizing problems as tasks, with the following definition of solvability:

\begin{definition}[Stabilizing Task Solvability] 
    A protocol \Prot solves a stabilizing task $T= \langle\mathcal{I}, \mathcal{O}, \task \rangle$ iff for any run $\config[0] : \sigma$, there exists an output configuration $(p_i, \conout_i)_{i=1}^n \in \task(\config[0])$ and a stabilization round $s_i$ for every process $p_i$, such that:
    \begin{equation*}
        \forall r \geq s_i: \delta_{\mathcal{P}}(view_{\config[0] : \sigma}(p_i,r)) =  \conout_i \textrm{ where $\conout_i$ denotes the output in $(p_i, \conout_i)_{i=1}^n$}
    \end{equation*}
    A protocol that solves a stabilizing task is called a stabilizing protocol.
\end{definition}

\subsection{Consensus and stabilizing consensus}
Consensus is traditionally defined via the following three properties.
\begin{enumerate}
    \item[\textbf{1.}] \textbf{Termination:}
        Any process $p_i$ eventually irrevocably decides on an output value $\conout_i$ at a round $f_i$.
    \item[\textbf{2.}] \textbf{Validity:}
        If a process $p_i$ decides $\conout_i$, then $\conout_i$ was the input value of some process $p_j$, i.e $\conout_i = \conin_j$.
    \item[\textbf{3.}] \textbf{Agreement:}
        If process $p_i$ decides on an output value $\conout_i$, and process $p_j$ decides $\conout_j$, then $\conout_i = \conout_j$.
\end{enumerate}

Note that we assume that $I$ is finite and that any combination of input configurations $\config[0]\in I^{\Pi}$ is possible. Therefore, the above conditions do not define a single unique task, but rather a family of tasks. For instance, binary consensus (where $|I|=2$), multi-valued consensus (where $|I|>2$) correspond to different tasks, while both are instances of consensus.

Stabilizing consensus is defined by replacing the termination property \textbf{1.} by the weaker condition \textbf{1'.}. Since stabilizing protocols are allowed to decide at each round, the agreement condition \textbf{3.} also needs to be adjusted to \textbf{3.'}.

\begin{enumerate}
    \item[\textbf{1'.}] \textbf{Stabilization:} Any process $p_i$ eventually stabilizes on some valid output value $\conout_i$ after some round $s_i$.
    \item[\textbf{3'.}] \textbf{Stable Agreement:}
        If process $p_i$ \emph{stabilizes} on an output value $\conout_i$, and another process $p_j$ \emph{stabilizes} on $\conout_j$, then $\conout_i = \conout_j$.
\end{enumerate}

Note that a task does not consider any termination or stabilization condition. Therefore, the stabilizing consensus problem, as well as terminating consensus, is determined by a consensus task and a stabilizing (respectively terminating) condition.

\begin{definition}[Consensus Task]\label{def:taskcons}
    We define the consensus task with respect to an input configuration set $\mathcal{I}$ as $T_{con(\mathcal{I})} := \langle \mathcal{I}, \mathcal{O}_{con(\mathcal{I})}, \task_{con(\mathcal{I})}\rangle$.
    
    Possible output configurations are vectors satisfying agreement
    $\mathcal{O}_{con(\mathcal{I})} := \{ (p_i,\conout_i)_{i=1}^n \; \conout_i \; v \in I(\mathcal{I}) \}$, where $I(\mathcal{I})$ denotes the set of input values induced by the input configurations $\mathcal{I}$, and
    the task map preserves validity
    $\task_{con(\mathcal{I})}( (p_i, \conin_i)_{i=1}^n ) := \{ (p_i, \conout)_{i = 1}^n \;|\; \conout = \conin_j \text{ for } j\in\{1,\dots,n\}\}$.
\end{definition}

\section{The Lossy-Link, and Delayed Lossy-Link Message Adversaries}
\label{sec:lldll}

We will first introduce two particular message adversaries defined for only two processes, namely \emph{Lossy-Link} (denoted by $\LL$) and \emph{Delayed Lossy-Link} (denoted by $\DLL$). Lossy-Link was introduced by Santoro and Widmayer~\cite{SW89} and revisited in~\cite{SWK09,CGP15}, where it was shown that consensus is impossible even if at most a single message may be lost in every round. We will show that stabilizing consensus is solvable in the \LL but impossible in the \DLL. 

Lossy-Link consists of 2 processes that communicate through a bi-directional link that may lose at most one message per round. For readability purposes, throughout this section we will denote the set of processes as $\Pi = \{ \tp, \tq \}$. 

\begin{definition}[Lossy-Link message adversary]\label{def:ll}
    We define the \LL as\footnote{Throughout this paper, $\omega$ denotes the first infinite ordinal, it is very convenient for expressing regular infinite sequences in a compact way.}:
    \[ LL:= \{ \tpq, \tqp, \tpqp \}^\omega . \]
    
\end{definition}

Although the Lossy-Link message adversary prohibits solving terminating consensus~\cite{SW89}, it admits a simple stabilizing protocol that solves stabilizing consensus as it satisfies both that any \LL pattern has a non empty kernel, and any \LL pattern has a trivially bounded broadcast time. Hence, we provide a 2-process instance of the MinMax Algorithm, introduced in ~\cite{CM19:DC}, for stabilizing consensus. Since we adapt it to our full-information and flooding model, it suffices to provide a decision map $\delta_{MinMax}: Views ( \LL) \rightarrow O$,\\
$\delta_{MinMax}(view_\sigma(p_i,r)) := \underset{x \in V(view_\sigma(p_i,r))}{\max} \{\min (\leafs(x))\}$

We omit a proof (for a complete formal treatment see \cite{CM19:DC}), but provide a sketch. The case where both have identical inputs is trivial, so assume w.l.o.g. that $\conin_\tp < \conin_\tq$. We distinguish three cases, either (a) $\tq$ never hears from $\tp$, or conversely (b) $\tp$ never hears from $\tq$, or (c) both hear from each other. In (c) both eventually know both values after some round $r$. They then always choose $\conin_\tp$, as its the smallest value they know and as its the smallest value in any message they will receive, starting from round $r+1$.

Case (a) is only possible in $(\tqp)^\omega$, so $\tp$ always hears from and chooses $\conin_\tq$, as its the maximum heard of in the last round ($\tq$ also chooses $\conin_\tq$ by validity). Similarly, case (b) is only possible in $(\tpq)^\omega$, where after round $1$ the minimum $\tq$ has heard of is $\conin_\tp$, and $\tq$ will never receive a larger value, thus always chooses $\conin_\tp$ ($\tp$ chooses the same by validity).

Now consider a different setting, where the communication link between $\tp$ and $\tq$ is also allowed to drop both messages, but only for at most $k$ consecutive rounds for some fixed $k$. We call this model the \emph{Bounded-Delay Lossy-Link}, abbreviated as $\BDLL(k)$. Following our previous notation, $\tnone$ denotes the silent graph.

\begin{definition}[Bounded-Delay Lossy-Link message adversary]\label{def:kbdll}
    We define the $k$-Bounded-Delay Lossy-Link message adversary, denoted by $\BDLL(k)$, as: 
   \[ \BDLL(k) := ( \bigcup_{i=0}^k \{ \tnone^i \} \cdot \{ \tpq, \tqp, \tpqp \})^{\omega}. \]

   More generally, we define the Bounded-Delay Lossy-Link message adversary, as 
   \[ \BDLL:= \bigcup_{k\in\mathbb{N}} \BDLL(k) . \]
\end{definition}

Note that, as some communication is always guaranteed, $\BDLL(k)$ has a non-empty kernel for any $k$, satisfying (i) and trivially (ii). This also applies to the $\BDLL$, where any sequence is an instance of some $\BDLL(k)$.

Nevertheless, just $\BDLL(1)$ already breaks $\delta_{MinMax}$, as it is not guaranteed anymore that the maximum in the system reaches all other processes in every round: Consider the graph sequence $( \tqp \comp \tnone )^{\omega}$ with $\conin_\tp < \conin_\tq$, where \tp always alternates its decision  value: it will choose $\conin_\tp$ if the last communication graph was $\tnone$, and $\conin_\tq$ if it was $\tqp$, as it only considers messages from the last round. In order to circumvent this alternation, we could adapt $\delta_{MinMax}$ to look back two rounds instead of only looking at the last one. In fact, we can generalize this idea, by adapting $\delta_{MinMax}$ to consider the past $k$ rounds and thus making it suitable for $\BDLL(k)$: $ \delta_{MinMax}^k(view_{\sigma}(p_i,r)) :=
    \underset{x \in V^k(view_{\sigma}(p_i, r))}{\max} \{ \min (\leafs(x)) \}$,
where $V^k(view_{\sigma}(p_i, r)) := \bigcup_{j = 0}^{min(k,r-1)} V(view_\sigma(p_i, r-j))$ is the set of views $p_i$ has received in the last $k$ rounds (when $k<r$).

A similar correctness argument as the one for $\delta_{MinMax}$ can be used to show that $\delta_{MinMax}^k$ solves stabilizing consensus in $\BDLL(k)$. However, any $\delta_{MinMax}^k$ will fail in $\BDLL(k+1)$. Charron-Bost and Moran cleverly fixed this in \cite{CM19:DC} and provided an algorithm capable of solving stabilizing consensus in $\BDLL(k)$ for \emph{any} fixed $k$, i.e., in $\BDLL$. Intuitively their safe MinMax algorithm works as follows: Any sequence $\sigma \in \BDLL$ is also a member of $\BDLL(k)$, for some $k$. Of course, one cannot choose some $l$ that is always larger than any $k$, but, as the sequence $\sigma$ is chosen beforehand and $k$ therefore fixed, one can gradually \emph{increase} $l$ such that it eventually surpasses any fixed $k$. And since $\delta_{MinMax}^l$ also correctly runs for any $k < l$, once $l$ surpassed $k$ it will run correctly forever. Note that we will formalize this in \cref{sec:iisdiis} for $n$ processes.

Nevertheless, we will prove that \emph{no stabilizing protocol} can exist if there is \emph{no fixed bound} $k$ on the length of consecutive silence periods. We call this model the \emph{Delayed Lossy-Link} message adversary, which can be seen as the limit of the \BDLL.

\begin{definition}[Delayed Lossy-Link Message Adversary]\label{def:dll}
    We define the Delayed Lossy-Link Message Adversary, denoted as \DLL as:
    \[ \DLL := (\tnone^* \cdot \{ \tpq, \tqp, \tpqp \})^\omega , \]

    where $\tnone^*$ is the Kleene star of $\tnone$, i.e.,
    $\tnone^* = \bigcup_{i \in \N\cup\{ 0\} }\{ \tnone^i\}$
\end{definition}
Note that the \DLL again has a non-empty kernel, as some communication always happens. It hence satisfies (i), does not satisfy (ii), as there is no bound on how many consecutive silent rounds may happen.

Also note that any communication pattern $\sigma\in\DLL$ is essentially a sequence in the \LL with arbitrarily long but finite periods of no communication in-between. Consequently, any communication pattern $\sigma\in\DLL$ has a unique corresponding silence-free communication pattern $\mu \in \LL$ where $\mu\sqsubseteq\sigma$, which we call the \emph{silence-free core} of $\sigma$. Reciprocally, we say that $\sigma$ is a delayed pattern of $\mu$.

\subsection{Stabilizing consensus is impossible in the DLL}\label{subsec:stabimp}

In this sub section, we prove one of our main results, namely the impossibility of stabilizing consensus in the \DLL.

The proof strategy is the following: we assume by contradiction that there is a stabilizing protocol that can solve consensus in the \DLL. Then we show that it is possible to construct a run that infinitely often decides for a conflicting configuration, i.e., with different decision values for \tp and \tq. The crucial argument for constructing a conflicting run is \cref{lem:patprot}, which proves that any stabilizing protocol in the \DLL can be translated into a patient protocol. A patient protocol allows processes to change their decision value only in rounds where they receive a message from another process.

We start by defining conflicted prefixes, which are somewhat similar (but different) to bivalent prefixes: whereas a bivalent prefix is a prefix that may eventually lead to mutually incompatible output configurations, a conflicted prefix is one that currently outputs an invalid configuration. For instance, the communication pattern $\tpq^\omega \in LL$ under $\delta_{MinMax}$  with $\conin_\tp > \conin_\tq$ yields a bivalent prefix for any round $r$. Nevertheless, $\delta_{MinMax}$ never outputs a conflicted prefix after round $1$ in this setting.

\begin{definition}[Conflicted prefix] \label{def:conflpref}
    Let $\mathcal{P}$ be a stabilizing protocol,  $T = \langle \mathcal{I},\mathcal{O} , \task \rangle$ a task, and $\sigma = (G_i)_{i \in \mathbb{N}}$ a communication pattern. We say that a prefix $\sigma^{k}$ is \emph{conflicted} with input configuration $\config[0] = (p_i, \conin_i)_{i=1}^n$, iff $\delta_{\mathcal{P}}(\config[r]_\sigma ) \notin \task(\config[0])$.
\end{definition}

Respectively, we say that a communication pattern $\sigma$ is \emph{conflicted infinitely often} iff there is a strictly increasing sequence $s: \mathbb{N} \rightarrow \mathbb{N}$ such that $\sigma|_1^{s(i)}$, the $s(i)$-prefix of $\sigma$, is conflicted. Clearly, a stabilizing protocol solves a task under a message adversary $\mathcal{M}$ iff there are no communication patterns in $\mathcal{M}$ that are conflicted infinitely often.

Note that the possibility of arbitrary long silence periods adds a new layer of uncertainty, particularly in the context of stabilizing protocols. In fact, this removes any possibility of acquiring any useful information through silence. For instance, in \LL, even if process $\tp$ does not receive a message in round $r$, it is at least sure that $\tq$ will receive a message in round $r$. The effect of this uncertainty on stabilizing protocols is reflected by the following definitions, \cref{lem:patience}, and \cref{lem:patprot}. 

Let $\pi$ be a prefix of length $k$. For notational simplicity, we will use $\delta_{p_i}(\pi)$ to denote $\delta_{\mathcal{P}}(view_{\pi} (p_i, k))$, whenever the protocol $\mathcal{P}$ is either unique or not immediately relevant to the context; similarly $\delta(\config[0] : \pi)$ denotes the output configuration $(p_i,\delta_{p_i}(\config[0] : \pi))_{i=1}^n$. Furthermore, if $\sigma$ is an infinite communication pattern and $\Prot$ is a stabilizing protocol, we denote by $\decfinalp{\config[0]: \sigma}{p_i}$ the stable decision by $p_i$ under $\sigma$ with initial configuration $\config[0]$.

\begin{definition}[Patient prefix] \label{def:patpref}
     Let $\mathcal{P}$ be a stabilizing protocol, and $T = \langle \mathcal{I},\mathcal{O} , \task \rangle$ a stabilizing task in the \DLL model. We say that a prefix $\pi$ is a \emph{patient} prefix for input configuration $\config[0]$, iff for every $r>0$; 
      $\delta(\config[0] : \pi ) = \delta(\config[0]: \pi \comp \tnone^r).$
\end{definition}

\begin{lemma}[Patience Lemma]\label{lem:patience}
    Let $\mathcal{P}$ be a stabilizing protocol for a task $T = \langle \mathcal{I},\mathcal{O} , \task \rangle$, $\config[0] \in \mathcal{I}$, and $\pi$ any prefix in the \DLL model. There exists $k \geq 0$ such that $\config[0] : \pi \comp \tnone^k$ is a patient prefix.
\end{lemma}

\begin{proof}
    Note that $\pi \comp \tpq^{\omega}$ is an admissible execution in the \DLL message adversary. Since $\Prot$ is a stabilizing algorithm, $\tp$ must eventually stabilize on a decision $\decfinalp{\config[0]: \pi \comp \tpq^{\omega}}{\tp}$. Hence there is some minimal $r_1 \geq 0$ such that for any $r> r_1$,
     $ \delta_{\tp}(\config[0] : \pi\cdot \tpq^{r}) =
        \decfinalp{\config[0]: \pi \comp \tpq^{\omega}}{\tp}.$
    
    Reciprocally, there is some minimal $r_2 \geq 0$ such that for any $r>r_2$, \\ $\delta_{\tq}(\config[0] : \pi\cdot \tqp^{r}) = 
        \decfinalp{\config[0] : \pi \comp \tqp^{\omega}}{\tq}.$

    Note that $\config[0] : \pi\cdot \tpq^{r_1}$ is indistinguishable from $\config[0] :\pi\cdot \tnone^{r_1}$ to $\tp$. Symmetrically  $\config[0] : \pi\cdot \tqp^{r_2}$ is indistinguishable from $\config[0] : \pi\cdot \tnone^{r_2}$ to $\tq$.

    Let $k = \max\{r_1,r_2\}$, and consider $\pi'= \config[0] : \pi \comp \tnone^k$. For any $k'$, $\pi'\comp \tnone^{k'}$ is indistinguishable from $\config[0] : \pi \comp \tpq^{k+k'}$ to $\tp$, and $\pi'\comp \tnone^{k'}$ is indistinguishable from $\config[0] : \pi \comp \tqp^{k+k'}$ to $\tq$.
    
    Moreover, since $k+k'\geq k$, it follows that $\decfinalp{\config[0] : \pi \comp \tpq^{\omega}}{\tp} = \delta_{\tp}( \pi' \comp \tnone^{k'}) = \delta_{\tp}(\pi')$, and $\decfinalp{\config[0] : \pi \comp \tqp^{\omega}}{\tq} = \delta_{\tq}( \pi' \comp \tnone^{k'}) = \delta_{\tq}(\pi')$. Consequently, $\pi' = \config[0] : \pi \comp \tnone^k$ is indeed a patient prefix.
\end{proof}

Note that \cref{lem:patience} allows us to define the \emph{patience} of a prefix $\pi$, for a particular stabilizing protocol $\Prot$ in the \DLL and a particular input configuration $\config[0]$.

\begin{definition}[Patience]\label{def:patience}
   Let $\Prot$ be a stabilizing protocol, $\config[0] \in \mathcal{I}$ be an input configuration, and $\pi$ a prefix in the \DLL. We say that $k \in \mathbb{N}$ is the patience of $\pi$ for $\Prot$ and $\config[0]$, if $k$ is minimal such that $\config[0] : \pi \comp \tnone^k$ is a patient prefix.
\end{definition}

\begin{definition}[Patient Protocol]\label{def:patprot}
    Let $\Prot$ be a stabilizing protocol in the \DLL. We say that $\Prot$ is patient if the patience of any prefix $\pi$ is $0$.
\end{definition}

\begin{lemma}[Patient Protocol Reduction]\label{lem:patprot}
    Let $\Prot$ be a protocol that solves a stabilizing task \\ $T = \langle \mathcal{I}, \mathcal{O}, \task \rangle$ in the 
    \DLL. Then there exists a patient protocol $\Prot'$ that solves $T$.
\end{lemma}

\begin{proof}
    Consider a communication pattern $\sigma = (G_i)_{i=1}^\infty \in \DLL$. We define inductively a prefix sequence $\pi_i$ in the following way: $\pi_1 := G_1 \comp \tnone^{\alpha_1}; \; \pi_{k+1} := \pi_k \comp G_{k+1} \comp \tnone^{\alpha_{k+1}}$, where $\alpha_{1} := patience(G_1); \; \alpha_{k+1} := patience(\pi_k \comp G_{k+1}).$ We define a communication pattern $\sigma' := (H_i)_{i=1}^\infty$, as the limit of $\pi_i$. 
    

    Consider a protocol $\mathcal{P}'$ with decision map $\delta'$ defined by: $\delta'_x(\sigma|_1^k)) = \delta_x(\pi_k)$. We will show that $\mathcal{P}'$ is a patient stabilizing protocol that solves $T$.
   
    By construction, each $\pi_i$ is a patient prefix of $\sigma'$, hence $\mathcal{P}'$ is patient. Note that $\sigma$ is an infinite sub-sequence of $\sigma'$ that does not include an infinite silence suffix $\tnone^{\omega}$, $\sigma'$ does not include an infinite silent suffix $\tnone^{\omega}$ either. Thus, $\sigma' \in DLL$. 

    Since $\mathcal{P}$ is a stabilizing protocol that solves $T$ in $\DLL$, in particular, $\delta$ provides a stabilizing solution for communication pattern $\sigma'$. Since $\delta(\config[0]:\sigma') = \delta'(\config[0]:\sigma)$, it follows that $\mathcal{P}'$ is a stabilizing protocol that solves $T$ for communication pattern $\sigma$.
    
    
    

\end{proof}

\begin{theorem}[Stabilizing consensus Impossibility in \DLL]\label{thm:stabagreeimp}
    Let $\Prot$ be an arbitrary stabilizing protocol in \DLL. There exists a valid communication pattern $\sigma\in\DLL$ such that $\Prot$ does not solve stabilizing consensus for $\sigma$.
\end{theorem}
    
\begin{proof}
    Let us assume for a contradiction that there is a stabilizing protocol $\mathcal{P}$ that solves stabilizing consensus in \DLL. By \cref{lem:patprot}, there exists a patient protocol $\Prot'$ that solves stabilizing consensus in \DLL. We will provide an inductive construction of a run that is conflicted. More precisely, we define a sequence of conflicted prefixes $\pi_i$ such that $\pi_{i} \subset \pi_{i+1}$, and $\lim_{i \rightarrow \infty} \pi_i = \sigma \in \DLL$.

    \textbf{Base case (Empty prefix $\mathcal{\varepsilon}$):} For the base case, consider an initial input configuration $\config[0]$ with different values $\conin_{\tp} \neq \conin_{\tq}$. Due to the validity condition, each process must decide on its own value. Since $\conin_{\tp} \neq \conin_{\tq}$, $\config[0]:\varepsilon$ is a conflicted prefix.

    \textbf{Induction step (Prefix $\pi_k$):} Assume that $\config[0]:\pi_k$ is a conflicted prefix, i.e., $ \delta_{\tp} (\config[0]:\pi_k) \neq \delta_{\tq} (\config[0]:\pi_k ).$
    

    We assert that at least one of $\config[0] : \pi_k \comp \tpq$, $\config[0] : \pi_k \comp \tpqp$,  $\config[0] : \pi_k \comp \tqp$ is conflicted. Assuming that none is conflicted, we derive a contradiction (as illustrated in \cref{fig:dllimp}).

    Recall that $\Prot'$ is patient and $\config[0]: \pi_k \comp \tpq$ is indistinguishable from $\config[0]: \pi_k \comp \tnone$ for $\tp$. Therefore, $\delta_{\tp}(\config[0]: \pi_k \comp \tpq) = \delta_{\tp}(\config[0]: \pi_k \comp \tnone) = \delta_{\tp}(\config[0]: \pi_k)$. 

    Since we assumed that $\config[0]: \pi_k \comp \tpq$ is not conflicted, then\\
    $\delta_\tp (\config[0]: \pi_k \comp \tpq) = \delta_\tq (\config[0]: \pi_k \comp \tpq)$. 

    Note that $\config[0]: \pi_k \comp \tpq$ is indistinguishable from $\config[0]: \pi_k \comp \tpqp$ for $\tq$, and therefore $\delta_{\tq}(\config[0]: \pi_k \comp \tpq) = \delta_{\tq}(\config[0]: \pi_k \comp \tpqp)$.

    Similarly, since $\config[0]:\pi_k \comp \tpqp$ is not conflicted by assumption,\\
    $\delta_\tq (\config[0]:\pi_k \comp \tpqp) = \delta_\tp (\config[0]:\pi_k \comp \tpqp)$.

    Note that $\config[0]:\pi_k \comp \tpqp$ is indistinguishable from $\config[0]: \pi_k \comp \tqp$ for $\tp$ and thus,\\
    $\delta_\tp (\config[0]:\pi_k \comp \tpqp) = \delta_\tp(\config[0]: \pi_k \comp \tqp)$.

    Since $\config[0]:\pi_k \comp \tqp$ is not conflicted by assumption, $\delta_\tp(\config[0]: \pi_k \comp \tqp) = \delta_\tq(\config[0]: \pi_k \comp \tqp)$.

    Finally, since $\config[0]:\pi_k \tqp$ is indistinguishable from $\config[0]:\pi_k \comp \tnone$ for $\tq$ and $\Prot'$ is patient, $\delta_\tq(\config[0]: \pi_k \comp \tqp) = \delta_\tq(\config[0]: \pi_k \comp \tnone) = \delta_\tq(\config[0]: \pi_k)$.

    Therefore $\delta_\tp (\config[0] : \pi_k) = \delta_\tq(\config[0]:\pi_k)$, which contradicts the induction hypothesis, namely that $\config[0]:\pi_k$ is conflicted. Thus, either one of $\config[0]:\pi \comp \tpq$, $\config[0]:\pi \comp \tpqp$, or $\config[0]:\pi \comp \tqp$ is a conflicted prefix.
    
    Choosing $G_{k+1} \in \{ \tpq, \tpqp, \tqp \}$ such that $\config[0]:\pi_{k} \comp G_{k+1}$ is conflicted, we define $\pi_{k+1} := \pi_k \comp G_{k+1}$. Setting $\sigma = \lim_{k\rightarrow\infty}\pi_k$, note that $\sigma \in \DLL$ by construction and each $\pi_k$ is a conflicted prefix: Therefore $\sigma$ is indeed conflicted infinitely often, and $\Prot'$ does not solve stabilizing consensus.
 
    \begin{figure}[ht]
        \begin{tikzpicture}
            \pgfmathsetmacro{\x}{0}
            \pgfmathsetmacro{\xstep}{2.6}
            
            \node at (\x,  0) (r1) {$\delta_{\tp}(\pi_k \comp \tnone)$};
            \pgfmathsetmacro{\x}{\xstep+\x}
            \node at (\x, 0) (r2) {$\delta_{\tp}(\pi_k \comp \tpq)$};
            \node at (\x, -1) (r3) {$\delta_{\tq}(\pi_k \comp \tpq)$};
            \pgfmathsetmacro{\x}{\xstep+\x}
            \node at (\x, -1) (r4) {$\delta_{\tq}(\pi_k \comp \tpqp)$};
            \node at (\x, -2) (r5) {$\delta_{\tp}(\pi_k \comp \tpqp)$};
            \pgfmathsetmacro{\x}{\xstep+\x}
            \node at (\x, -2) (r6) {$\delta_{\tp}(\pi_k \comp \tqp)$};
            \node at (\x, -3) (r7) {$\delta_{\tq}(\pi_k \comp \tpq)$};
            \pgfmathsetmacro{\x}{\xstep+\x}
            \node at (\x, -3) (r8) {$\delta_{\tq}(\pi_k \comp \tnone)$};
        
            \foreach \i in {1, 3, 5, 7} {
                \pgfmathtruncatemacro{\j}{\i+1}
                \draw[draw=none] (r\i.east) -- (r\j.west)
                    node[midway,circle,inner sep=0pt] (edge{\i}{\j}eq) {$=$};
            }
    
            \pgfmathsetmacro{\x}{\x+1.5}
            \draw[->] (-1.5,0) -- (r1);
            \draw (-1.5,0) -- (-1.5,-3.5) -- (\x,-3.5)
                node[midway,draw,circle,inner sep=2pt,fill=white] (edgeneq) {$\neq$} -- (\x,-3);
            \draw[->] (\x,-3) -- (r8);
    
            \begin{scope}[on background layer]
                \node[fit=(r2) (r3),fill=lightestgray,rounded corners,inner sep=-1pt] {};
                \node[fit=(r4) (r5),fill=lightestgray,rounded corners,inner sep=-1pt] {};
                \node[fit=(r6) (r7),fill=lightestgray,rounded corners,inner sep=-1pt] {};
            \end{scope}
        \end{tikzpicture}
        \caption{An illustration of the induction step in \cref{thm:stabagreeimp}. Decision values grouped by a grey rectangle correspond to decision values of the same prefix, i.e., they should match otherwise we have a conflicted prefix. Decision values connected via an equality sign are necessarily identical because the process cannot distinguish the two prefixes. This creates a chain of equalities that is broken by our induction hypothesis, namely that $\pi_k \comp \tnone$ is conflicted, implying one of the grey boxes, and consequentally a longer prefix, must be conflicted as well.}
        \label{fig:dllimp}
    \end{figure}
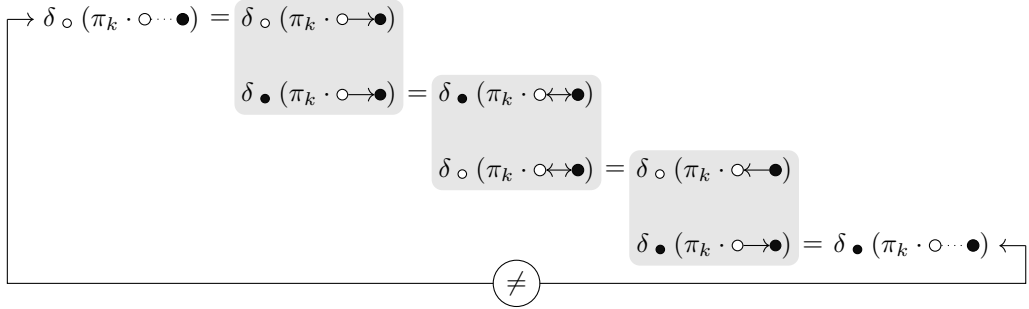
\end{proof}

\section{The Iterated Immediate Snapshot and the Lossy Iterated Immediate Snapshot Models}
\label{sec:iisdiis}

Previously, we introduced the Lossy-Link and Delayed Lossy-Link message adversary as a starting point for our stabilizing consensus impossibility result. This is for two main reasons: first, the simplicity of the Lossy-Link highlights the source of the impossibility result without adding distracting details to the proofs, and second, the \DLL impossibility translates transparently to other message adversaries that generalize \DLL.

One such message adversaries is the \emph{Iterated Immediate Snapshot} message adversary, which is computationally equivalent to the asynchronous Iterated-Immediate Snapshot model defined as an $n$-process asynchronous shared memory model where processes communicate via atomic snapshots. The first approach towards this model was introduced by Borowski and Gafni~\cite{BG93}, who presented the \emph{Immediate-atomic-snapshot model} as a generalization of the classical FLP model~\cite{FLP/ACM}. They also re-used it as a convenient model for characterizing wait-free asynchronous shared memory systems~\cite{BG97}. 

In the \IIS model, one considers a set $\Pi = \{ p_i\}_{i=1}^{n}$ of processes that communicate through a shared memory object $M$. Each process may write into its own register $M[i]$ and is able to read the whole memory $M$ in an atomic snapshot operation. A process $p_i$ that executes an atomic snapshot writes into its register $M[i]$ and reads all of $M$ instantaneously. Since snapshot operations executed by different processes may occur concurrently, it is assumed that every process in a concurrent set is capable of reading each others output. Although atomic snapshot objects appear to be very strong, it has been shown that they are equivalent to single-reader single-writer wait-free shared registers~\cite{HKR14}. 

Perhaps one of its most useful features of the \IIS model is that it allows us to express an asynchronous shared memory model in terms of a synchronous message adversary~\cite{AG13}, an asynchronous snapshot by all processes, can hence be represented by a communication graph. If a process $\p$ and a process $\q$ execute a snapshot operation, there are $3$ possibilities, either $\p$ executed its snapshot before $\q$, $\q$ before $\p$, or both snapshots are concurrent. If $\p$ executes its snapshot before $\q$, then $\p$ cannot read what $\q$ wrote, but $\q$ will be able to read the value from $\p$. Reciprocally, if $\q$ executes its snapshot before $\p$, $\q$ will see $\p$'s value but not vice versa. If $\p$ and $\q$ have concurrent snapshots, then both are able to read each other's value. Since we consider time to be linearly ordered, and snapshots occur within the same time frame (round), then this means that for any pair of processes $\p$ and $\q$, either $\p$ reads $\q$ or $\q$ reads $\p$, or both. When translated to communication graphs, this implies that each communication graph in the \IIS model is \emph{semi-complete} \footnote{A directed graph is semi-complete if for any pair of vertices $v,w$, either $(v,w) \in E(G)$ or $(w,v) \in E(G)$.}. Since the time at  which each process executed its snapshot determines completely the memory registers that each process is able to read, every communication graph in the \IIS model is also \emph{transitive} \footnote{A directed graph $G$ is transitive iff $(v,w) (w,z) \in E(G) \Rightarrow (v,z) \in E(G)$.}. 

The converse can also be shown, such that for every communication pattern $\sigma$ of semi-complete and transitive graphs, there is an \IIS schedule that matches $\sigma$. Thus, for the rest of this paper, we will simply consider that $\IIS := \mathcal{G}^\omega$, where $\mathcal{G}$ is the set of semi-complete and transitive graphs on $\Pi$. Note carefully that \LL satisfies both transitivity and semi-completeness.

\subsection{The Lossy Iterated Immediate Snapshot Model}

Having introduced the \IIS model as an $n$-process generalization of the Lossy-Link model, we will now add read-omission faults to the immediate snapshot operations. This means that, whenever processes execute a snapshot in a round, there may be at most $f$ total read-omission failures of other processes' registers.

Note that the empty graph $\tnone \in \DLL$, can be viewed as a read-omission in either $\tpq$ or $\tqp$. However, for $n \geq 3$, faulty snapshots are slightly more complex than simply considering empty graphs. Consequently, we define the set of \emph{lossy snapshot graphs}, $\Phi(f)$ as the set of directed graphs from $\mathcal{G}$ (consisting of semi-complete and transitive directed graphs on $\Pi$) with at least one missing edge and at most $f\leq n (n-1)$ missing edges. More precisely, $\Phi(f) := \{ G \setminus F \; \vert \;  G \in \mathcal{G}, F \subseteq E(\mathcal{G}),1 \leq \vert F \vert \leq f\}$.

Similarly to \cref{sec:lldll}, we start by limiting the number of consecutive iterations where a snapshot might be faulty. We denote the \emph{Bounded Lossy Iterated Immediate Snapshot} model by $\BLIIS(f,k)$, where $f \leq n(n-1)$ is the maximum number of omission faults per iteration, and $k$ is the maximum number of consecutive rounds an immediate snapshot might be faulty. We then relax this restriction by requiring that correct snapshot iterations must only happen infinitely often (similar to the \DLL). We call this model the \emph{Lossy Iterated Immediate Snapshot} model, denoted by $\LIIS(f)$. Note that for $2$ processes, $\BLIIS_{\{\tp, \tq\}}(1,k)$ corresponds to $\BDLL(k)$, while $\LIIS_{\{\tp, \tq\}}(1)$ corresponds to \DLL.
 




These message adversaries can be formally defined as follows:

\begin{definition}[\IIS, \BLIIS, \LIIS Message Adversaries]\label{def:iisbdiisdiis}
    Let $\mathcal{G}$ be the set of semi-complete and transitive graphs on $\Pi$, and $\Phi(f)$ the set of faulty snapshot graphs for $f \leq n(n-1)$.

    \begin{itemize}
        \item $\IIS := \mathcal{G}^{\omega}$
        \item $\BLIIS(f,k) := (\{ \Phi(f)^i \}_{i=0}^k \comp \IIS) ^\omega$
        \item $\BLIIS(f) := \bigcup_{k \in \mathbb{N}} \BLIIS(f,k)$
        \item $\LIIS(f) := (\Phi(f)^* \comp \IIS)^\omega$, where $\Phi(f)^*$ is the Kleene star of $\{ \Phi(f) \}$.
    \end{itemize}
\end{definition}

First, we prove that all presented message adversaries have a non-empty kernel, and hence satisfy (i).

\begin{lemma}\label{lem:bliisnek}
    \IIS, $\BLIIS(f, k)$, $\BLIIS(f)$ and $\LIIS(f)$ all have a non empty kernel, i.e.,  for any $\sigma$, we have $Ker(\sigma) \neq \emptyset$.
\end{lemma}

\begin{proof}
    As every communication graph in $\sigma\in\IIS$ is semi-complete and transitive, $\sigma(r)$ contains at least one process that directly reaches all other processes\footnote{It is a well known result that transitive tournaments are equivalent to strict total order graphs, and hence finite transitive tournaments always contain a dominating vertex. Thus, any finite semi-complete and transitive graph also includes a dominating vertex. See for example~\cite{graphsanddigraphs} for more results on tournaments.}. As there are only finitely many processes, the pigeonhole principle guarantees that in the infinite sequence $\sigma$ at least one process reaches all others infinitely often, implying $Ker(\sigma)\neq\emptyset$.

    Now assume $\sigma = (G_i)_{i=1}^\infty$ to be an admissible run in either $\BLIIS(f, k)$, $\BLIIS(f)$ or $\LIIS(f)$. By construction, in any of these models, there is an infinite number of instances of graphs $G \in \IIS$ in $\sigma$, i.e., $(G_{s(i)})_{i=1}^\infty \sqsubseteq \sigma$ and $(G_{s(i)})_{i=1}^\infty \in \IIS$. But this is a sub-sequence with a non-empty kernel as proven above, therefore $Ker(\sigma)\neq\emptyset$ holds.
\end{proof}

We will now argue that $\IIS$, $\BLIIS(f,k)$ and $\BLIIS(f)$ allow solving stabilizing consensus. We start with the following technical lemma:

\begin{lemma}
    Let $\sigma$ be any communication pattern. There exists a round $s_i$ such that no process learns any new input values anymore. Formally,\\
    $\forall p_i\in\Pi : \forall r > s_i : \leafs(view_\sigma(p_i,r)) = \leafs(view_\sigma(p_i,s_i))$. We call this set $p_i$'s \emph{stable value set}.
\end{lemma}

\begin{proof}
    As the set of processes $\Pi$ is finite and input values are never forgotten in a run, the set $\leafs(view_\sigma(p_i,r))$ can be increased only finitely often.
\end{proof}

\begin{lemma}\label{lem:kerstabval}
    Let $\sigma$ be any communication pattern. After round $s_i$, $p_i$'s stable value set contains the stable value set of all processes in $Ker(\sigma)$. Formally, $ \forall p_i \in Ker(\sigma)$, $\forall p_i\in\Pi: \leafs(views_\sigma(p_j,s_j)) \subseteq \leafs(views_\sigma(p_i,s_i))$.
\end{lemma}

\begin{proof}
    By contradiction. Assuming the contrary implies that $p_i$ never receives a message from $p_j \in Ker(\sigma) $that reveals some input value known to $p_j$. This contradicts the assumption that $p_j$ is in the kernel.
\end{proof}

\begin{corollary}\label{cor:keridstabval}
    For any communication pattern $\sigma$ with $Ker(\sigma)\neq\emptyset$, processes in the kernel have identical stable value sets.
\end{corollary}

\begin{theorem}
    For any $\sigma \in \IIS \cup \BLIIS(f,k) \cup \BLIIS(f)$, 
    \begin{equation*}\label{eq:minmaxk}
      \delta_{MinMax}^{r/2}(view_{\sigma}(p_i,r)) :=
        \underset{x \in V^{r/2}(\config[0]:\sigma,p_i, r)}{\max} \{ \min (\leafs(x)) \}
    \end{equation*}
    solves stabilizing consensus in $\sigma$. Note that the number of past rounds considered (i.e, $\frac{r}{2}$) now depends on the current round $r$.
\end{theorem}

\begin{proof}
    We set $s_g = \max\{s_1, \dots, s_n\}$, i.e., the round after which all processes have arrived at their stable value sets. From \cref{lem:kerstabval}, it follows that for all rounds $r > s_g$,\\ $\forall p_j\in Ker(\sigma) : \forall p_i \in \Pi: \min(\leafs(view_\sigma(p_j,r))) \geq \min(\leafs(view_\sigma(p_i,r)))$. From \cref{cor:keridstabval}, we infer that processes in the kernel have identical minima.

    \begin{itemize}
        \item Assume $\sigma \in \IIS$, and $r-1 > s_g$. As any $\sigma(r)$ is semi-complete and transitive, all processes not in the kernel receive the view of a process $p_j \in Ker(\sigma)$. By the above reasoning, $p_j$'s minimum is larger than all other minima, therefore the decision function $ \delta_{MinMax}^{r/2}$ chooses correctly from round $s_g+1$ on.

        \item Assume $\sigma \in \BLIIS(f,k)$, and $r-k > s_g$. As $\sigma(r)$ is semi-complete and transitive for at least one graph among the ones for $r \in \{ r-k, \dots, r \}$, we can repeat the previous reasoning but applied to views received in the previous $k$ rounds, $V^k(\config[0]:\sigma,p_i,r)$. Eventually, i.e., after round $s_j+k$, we are sure that the interval $r-k$ to $r$ contains a view with the stable value set of a process $p_i$ in the kernel. Note that as $\IIS \subseteq \BLIIS(f,k)$, this also solves stabilizing consensus on the \IIS.

        \item Assume $\sigma \in \BLIIS(f)$, and $r > s_g$. As $\sigma$ does not have a predefined bound $k$, we cannot resort to the previous reasoning. However, we know there \emph{is} a bound, as $\sigma\in\BLIIS(f)$ implies $\sigma\in\BLIIS(f,k)$ for some $k$. As any $\delta^k$ solves stabilizing consensus also on any $\sigma\in\BLIIS(f,l)$ for $l<k$ by inclusion, we can dynamically consider the previous $\frac{r}{2}$ rounds. This ensures that (a) we eventually surpass the fixed bound $k$, and thus eventually always consider the view of a process in the kernel, and (b) eventually $r - \frac{r}{2}$ also exceeds $s_j$ and we only consider stable value sets of $p_j\in Ker(\sigma)$. Together with the fact that, $p_j$'s stable value set contains the largest minimum among all processes, we conclude that $\delta^{r/2}$ solves stabilizing consensus on $\sigma\in \IIS \cup \BLIIS(f,k) \cup \BLIIS(f)$.
    \end{itemize}
\end{proof}

As the main result of this section, we will now prove that stabilizing consensus is impossible in $\LIIS(1)$, i.e., in the presence of just one faulty register per round finitely consecutive rounds. We extend the proof of \cref{thm:stabagreeimp} from \DLL to $\LIIS(1)$ by identifying four admissible graphs in  $\LIIS(1)$ that ``replicate'' the \DLL impossibility result in the \LIIS model.

Note that $\IIS = LIIS(0) \subseteq \LIIS(1) \subseteq \ldots \subseteq \LIIS(f) \subseteq \LIIS(f+1)$, thus if stabilizing consensus is impossible in $\LIIS(1)$, then it is also impossible in $\LIIS(f)$ for any $f \geq 1$. 

For the remainder of this paper, we will consider a system with at least $n\geq 2$ processes, and focus on two distinct processes $p_1$ and $p_2$. For convenience, we assume that $\tp\mapsto p_1$, i.e., as the choice of index is free, we rename $\tp$ in $\DLL$ to $p_1$ in the context of the $\LIIS(1)$, and symmetrically $\tq\mapsto p_2$. We use the following convenient notations: $\Tilde{\Pi} := \Pi \setminus \{ p_1, p_2 \}$; $K_{\Tilde{\Pi}}$, is the complete graph on $\Tilde{\Pi}$; $\G_1 \oplus \G_2$ represents the graph defined by $V(\G_1 \oplus \G_2) := V(\G_1) \cup V(\G_2)$ , and $E(\G_1 \oplus \G_2) := E(\G_1) \cup E(\G_2) \cup V(\G_1) \times V(\G_2)$, i.e., adding all possible edges from $G_1$ to $G_2$ but not the other way round, and also extend it to runs $\sigma \oplus K_{\Tilde{\Pi}} := (G_i \oplus K_{\Tilde{\Pi}})_{i=1}^\infty$. 

We define the following graphs: $\Gpq := (p_1\pq p_2) \oplus K_{\Tilde{\Pi}}$; $\Gpqp := (p_1 \pqp p_2) \oplus K_{\Tilde{\Pi}}$; $\Gqp := (p_1 \qp p_2) \oplus K_{\Tilde{\Pi}}$ ; and $\Gnone := (p_1 \none p_2) \oplus K_{\Tilde{\Pi}}$. Note that \Gpq, \Gpqp and \Gqp are semi-complete and transitive and therefore snapshot graphs in the \IIS. Likewise, $\Gnone \in  \Phi(1)$, since it can be obtained by removing $(p_1,p_2)$ from \Gpq or by removing $(p_2,p_1)$ from \Gqp.

We show that any sequence in the \DLL can be extended to a sequence in the $\LIIS(1)$ where $p_1$ and $p_2$ have identical views, implying that a protocol solving stabilizing consensus in the $\LIIS(1)$ necessarily solves it in the \DLL, deriving a contradiction.

\begin{lemma}\label{lem:dllembedliis}
    Any $\sigma \in \DLL$ can be extended to a $\sigma'\in\LIIS(1)$, s.t., $Views(\sigma) = Views_{\{p_1, p_2\}}(\sigma')$.
\end{lemma}

\begin{proof}
    Let $\sigma \in \DLL$ and consider for $\sigma' = \sigma \oplus K_{\Tilde{\Pi}}$. By construction we know that\\
    $\sigma'(r) \in \{\Gpq, \Gpqp,\Gqp,\Gnone\}$ and as $\sigma$ does not contain an infinite silent suffix (i.e., $\sigma|_k^\infty \neq (p_1\none p_2)^\omega$ for all $k > 0$), $\sigma'$ also has no infinite suffix $(\Gnone)^\omega$ where no communication between $p_1$ and $p_2$ happens. Thus, $\sigma'$ is an admissible $LIIS(1)$ sequence.

    By construction, in $\sigma'$, $p_1$ and $p_2$ only receive messages from each other, moreover, they only do so in rounds where they also receive a message in $\sigma$. Therefore their views are identical $Views(\sigma) = Views_{\{p_1, p_2\}}(\sigma')$.
\end{proof}

\begin{theorem}[Stabilizing consensus Impossibility in $\LIIS(1)$]\label{thm:stabimpdiis}
    Let $\Prot$ be an arbitrary stabilizing protocol in the $\LIIS(1)$ model. $\Prot$ does not solve stabilizing consensus.
\end{theorem}

\begin{proof}
    Assume for a contradiction that $\Prot$ solves stabilizing consensus in the $\LIIS(1)$ model. Take any $\sigma \in \DLL$, from \cref{lem:dllembedliis}, it follows that there exists a $\sigma'\in\LIIS(1)$ with identical views for $p_1$, $p_2$. The decision map $\delta_{\Prot}$ therefore solves stabilizing consensus on $\sigma\in\DLL$ directly contradicting \cref{thm:stabagreeimp}.

\end{proof}

\section{Conclusion}\label{sec:conclusion}
    In this paper we provided the first stabilizing consensus impossibility result where the existence a non-empty kernel is guaranteed (i) but the number or rounds where everybody hears from a member in the kernel (ii) is not . While it has been shown that a non-empty kernel is necessary for solving stabilizing consensus in \cite{CM19:DC}, we prove that it is not sufficient.

    At the core of the stabilizing consensus impossibility lies the fact that padding the communication with arbitrarily long but finite silence periods greatly impairs the decision power of stabilizing protocols, essentially forcing them to fix a value during silence periods. This limitation enables us to find a valid prefix that has conflicting decision values, and thus prevents stabilization. This result highlights the importance of communication liveness within a system, since any bounded variants of the \DLL model are capable of solving stabilizing consensus.

    Furthermore, we extended this impossibility result to the Iterated Immediate Snapshot model, where the possibility of a single read-faulty snapshot makes stabilizing consensus impossible. This result sheds new light on the impact of omission faults on consensus, even when the termination condition is relaxed to stabilization, and the communication assumptions are as strong as a shared-memory model.

\bibliography{lit}

\appendix


\end{document}